\documentclass[twocolumn,showpacs,showkeys,prd,superscriptaddress]{revtex4-1}

\usepackage{array}
\usepackage{amsmath,amsfonts,amssymb,dsfont,mathrsfs,centernot}
\usepackage{graphicx}
\usepackage{subfigure}
\usepackage[colorlinks=True,linkcolor=blue,citecolor=blue]{hyperref}
\usepackage{tensor}
\usepackage{comment}

\allowdisplaybreaks

\usepackage{amsmath,amssymb,amsfonts,dsfont,mathrsfs,amsthm}
\usepackage{centernot}
\usepackage{hyperref}
\usepackage{xcolor}
\usepackage{comment}
\hypersetup{linktocpage,colorlinks=true,urlcolor=green,linkcolor=blue,citecolor=red}
\usepackage{feynmf}
\usepackage{siunitx}
\usepackage{array}
\usepackage{ulem}
\usepackage{tikz}
\usetikzlibrary{shapes,snakes}
\usepackage[framemethod=tikz]{mdframed}
\usetikzlibrary{calc}

\newtheorem{Thm}{Theorem}

\newcommand{\abs}[1]{\left|#1\right|}

\newcommand{\cdf}{{\boldsymbol{\mathcal{D}}}}
\newcommand{\covd}{\mathcal{D}}

\newcommand{\df}{{\bf{d}}}

\newcommand{\fy}{\centernot}

\newcommand{\ga}{\gamma}

\newcommand{\J}{\mathscr{J}}

\newcommand{\Lag}{\mathscr{L}}

\newcommand{\Ri}{\mathcal{R}}

\newcommand{\tor}{\mathcal{T}}
\newcommand{\w}{{\scriptstyle\wedge}}

\newcommand{\bps}{\ensuremath{\bar{\psi}}}

\newcommand{\Bps}{\ensuremath{\bar{\Psi}}}

\newcommand\VI[2]{\hat{e}^{\hat{#1}}_{\hat{#2}}}
\newcommand\VIF[1]{\hat{\bf{e}}^{\hat{#1}}}
\newcommand\VIN[2]{\hat{E}^{\hat{#1}}_{\hat{#2}}}

\newcommand\hvif[1]{\hat{\bf{e}}^{{#1}}}

\newcommand\GAM[1]{{\gamma}^{\hat{#1}}}
\newcommand\gam[1]{\gamma^{{#1}}}

\newcommand\PA[1]{\partial_{\hat{#1}}}

\newcommand\SPI[1]{\hat{\omega}_{\hat{#1}}}
\newcommand\SPIF[2]{\hat{\boldsymbol{\omega}}^{\hat{#1}}{}_{\hat{#2}}}

\newcommand\hspif[2]{\hat{\boldsymbol{\omega}}^{{#1}}{}_{{#2}}}
\newcommand{\RIF}[2]{\hat{\boldsymbol{\mathcal{R}}}^{\hat{#1}}{}_{\hat{#2}}}
\newcommand{\hRif}[2]{\hat{\boldsymbol{\mathcal{R}}}^{{#1}}{}_{{#2}}}

\newcommand{\TF}[1]{\hat{\boldsymbol{\mathcal{T}}}^{\hat{#1}}}

\newcommand{\hcont}[3]{\hat{\mathcal{K}}_{#1}{}^{#2}{}_{#3}}

\newcommand{\CONTF}[2]{\hat{\boldsymbol{\mathcal{K}}}^{\hat{#1}}{}_{\hat{#2}}}

\newcommand{\vev}[1]{\ensuremath{\left<#1\right>}}

\newcommand{\beq}{\begin{equation}}
\newcommand{\eeq}{\end{equation}}
\newcommand{\ber}{\begin{eqnarray}}
\newcommand{\eer}{\end{eqnarray}}

\renewcommand{\(}{\left(}
\renewcommand{\)}{\right)}
\renewcommand{\[}{\left[}
\renewcommand{\]}{\right]}

\newcommand*{\de}[1][]{\mathop{\mathrm{d}#1}\nolimits}

\hypersetup{
  pdftitle={Fermion Masses Through  Condensation in Spacetimes with Torsion},
  pdfauthor={Oscar Castillo-Felisola},
  pdfkeywords={Condensation,}  {Extra dimensions,} {Torsion,} {Generalised Gravity.},
  pdflang={English}
}

\begin{document}

\title{Fermion masses through  condensation in spacetimes with torsion}

\author{Oscar \surname{Castillo-Felisola}}
\email[Corresponding Author: ]{o.castillo.felisola@gmail.com}
\affiliation{Centro Cient\'\i fico Tecnol\'ogico de Valpara\'\i so, Chile.}
\affiliation{Departamento de F\'\i sica, Universidad T\'ecnica Federico Santa Mar\'\i a, Casilla 110-V, Valpara\'\i so, Chile.}

\author{Cristobal \surname{Corral}}
\email{cristobal.corral@postgrado.usm.cl}
\affiliation{Departamento de F\'\i sica, Universidad T\'ecnica Federico Santa Mar\'\i a, Casilla 110-V, Valpara\'\i so, Chile.}

\author{Cristian \surname{Villavicencio}}
\email{cristian.villavicencio@mail.udp.cl}
\affiliation{Instituto de Ciencias B\'asicas, Universidad Diego Portales, 
Casilla 298-V, Santiago, Chile.}

\author{Alfonso R. \surname{Zerwekh}}
\email{alfonso.zerwekh@usm.cl}
\affiliation{Centro Cient\'\i fico Tecnol\'ogico de Valpara\'\i so, Chile.}
\affiliation{Departamento de F\'\i sica, Universidad T\'ecnica Federico Santa Mar\'\i a, Casilla 110-V, Valpara\'\i so, Chile.}

\begin{abstract}
  In this paper we argue the possibility that  fermion masses, in particular quarks, originate through the condensation of a fourth family that interacts with all of the quarks via a contact four-fermion term coming from the existence of torsion  on the spacetime. Extra dimensions are considered to avoid the hierarchy problem.
\end{abstract}


\maketitle

\section{Introduction}

Recently, ATLAS and CMS experiments at the CERN Large Hadron Collider (LHC) found a signal consistent with the standard model Higgs boson, with an approximate mass of \SI{125.6}{\GeV}~\cite{Aaltonen:2012qt,Aad:2012tfa,Chatrchyan:2012ufa}. This discovery will shed light on the mechanism behind the electroweak symmetry breaking (EWSB). Although the  establishment of the quantum numbers of the discovered resonance is a pending task, it is crucial to determine whether the EWSB is produced by weak or strong  coupled dynamics.

The standard model (SM) of weak and strong interactions, has proved itself to be remarkably consistent with the experimental measurements, including the high-precision tests~\cite{PDG}. However, the lack of compatibility with the gravitational interaction has driven the community to believe that the SM is a low-energy effective framework  of a yet unknown fundamental theory. One of the problems that points in that direction is the {\it hierarchy problem}, which indicates that new physics should appear at a few \si{\TeV} in order to stabilize the Higgs mass at scales much lower than the Planck scale \SI{e19}[$\sim$]{\GeV}.

Alternatively, strong coupled  scenarios of EWSB could solve the hierarchy problem as long as no fundamental scalars turn nonperturbative above the electroweak (EW) scale, while the breakdown of the electroweak symmetry is caused by  condensed states in the vacuum~\cite{PhysRevD.41.1647,Hill:2002ap,Hirn:2007we,Hirn:2008tc,Belyaev:2008yj,Quigg:2009xr,Andersen:2011yj}. On several of these models the EW symmetry is broken through the condensation of fermions, generating  a composite scalar which acts as Higgs boson~\cite{PhysRevD.76.055005,PhysRevD.78.115010,Sannino:2009za}. Even if these theories are successful in breaking the EW symmetry, they should be extended for giving masses to fermions~\cite{ArkaniHamed:2002qx,ArkaniHamed:2002qy,PhysRevD.69.075002,Perelstein2007247,Schmaltz:2005ky,Dimopoulos1982206,Kaplan1984183,Banks1984125,Kaplan1984187,Georgi1984152,Georgi1984216,Dugan1985299,Giudice:2007fh,Csaki:2008zd,Contino:2010rs,Barbieri:2012tu,Sakai2013429,KerenZur2013394}.


Recently, a mechanism for breaking the EW symmetry through the condensation of a fourth family of quarks within the framework of extra dimensions has been proposed~\cite{Burdman:2007sx}. 

In this model, the condensation is mediated by the exchange of Kaluza-Klein gluons, while a four-fermion interaction is added in order to communicate with the SM sector. In the effective theory, the four-fermion interaction will give origin to the Yukawa interaction of the composite Higgs. Although the construction of the four-fermion term is based on  symmetry and universality arguments, it has still been arbitrary. 
In this respect, the situation is similar to the SM where the Yukawa couplings are arbitrary and unrelated to the gauge sector.


Although this model gives origin to  masses and mixing on the quark sector due to the underlying four-fermion interaction on the bulk, a good reproduction of the CKM matrix requires  certain level of nonuniversality~\cite{CarcamoHernandez:2012xy}. 

The aim of this paper is to study the possible gravitational origin of the four-fermion interaction, in the context of the Cartan-Eintein theory in five dimensions, where the presence of torsion gives rise naturally to a  term with the desired characteristics. In this type of scenario, extra dimensions are considered because in four dimensions the gravitational scale --Planck's mass-- is huge, and phenomenological effects are heavily suppressed (see for example Ref~\cite{Chang:2000yw,Kostelecky:2007kx,Zubkov:2010sx}).

This paper is organized as follows. In Sec.~\ref{CEG} a review about the induction of a four-fermion contact interaction from the coupling of Cartan-Einstein gravity with fermions is shown. In Sec.~\ref{sec:FC} a brief derivation of the effective model is presented and  the fermion condensation of the set up is performed. Finally, a discussion of results and conclusions are given in Sec.~\ref{sec:disc}. For completeness, a series of appendixes have been included: In Appendix~\ref{app:notation} the notation is explained; In Appendix~\ref{sec:actions} the equivalency between the gravitational and Dirac's actions written in differential forms and their customary form is proven. Additionally, in Appendix~\ref{sec:Fierz} a set of useful Fierz identities is stated.










\section{\label{CEG}Cartan-Einstein Gravity}

Cartan generalized the gravitational theory of Einstein, by considering connections which are not necessarily torsion free. This generalization is easily worked out using  the first order formalism of gravity. 

The information of the spacetime geometry is then encoded in a pair of fields, the vielbein, defined through
\begin{align}
  \hat{g}_{\hat{\mu}\hat{\nu}} = \eta_{\hat{a} \hat{b}}\VI{a}{\mu}\VI{b}{\nu},
\end{align}
and the spin connection, $\(\hat{\omega}_{\hat{\mu}}\)^{\hat{a}}{}_{\hat{b}}$, which encipher the same information as the Levi-Civita connection plus an additional term referring to the presence of torsion.

After defining the vielbein 1-forms, $\VIF{a} = \VI{a}{\mu}\df x^{\hat{\mu}}$, and the spin connection 1-form, $\SPIF{a}{b} = \(\hat{\omega}_{\hat{\mu}}\)^{\hat{a}}{}_{\hat{b}}\df x^{\hat{\mu}}$, the curvature of the spacetime is found through the structure equations
\begin{align}    
  \df\VIF{a}+\SPIF{a}{b}\w\VIF{b} &= \TF{a},\label{struc.eq.1}\\   
  \df\SPIF{a}{b}+\SPIF{a}{c}\w\SPIF{c}{b}&=\RIF{a}{b}.\label{struc.eq.2}
\end{align}
where $\TF{a}$ and $\RIF{a}{b}$ are the torsion and curvature 2-forms, respectively.

Finally, the gravitational action is 
\begin{align}
  S_{\text{gr}} = \frac{1}{2\kappa^2}\int\frac{\epsilon_{\hat{a}_1\cdots \hat{a}_D}}{(D-2)!}\hRif{\hat{a}_1 \hat{a}_2}{}\w\hvif{\hat{a}_3}\w\cdots\w\hvif{\hat{a}_D}.\label{CE-action}
\end{align}
The difference between this action and the one of Einstein-Hilbert is that the curvature tensor has contributions due to the torsion.

The action shown in Eq.~\eqref{CE-action} is the minimal extension of gravitation due to torsion. More general theories can be built out of curvature and torsion; however, there are ambiguities on the procedure, which are bypassed by restricting oneself to the minimal construction~\cite{RevModPhys.48.393,Belyaev:1998ax,Shapiro:2001rz,PhysRevD.75.034014}.

When considering pure gravity, the equation of motion from Eq.~\eqref{CE-action} are the usual Einstein's equations, because the equations of motion for the spin connection implies a vanishing torsion. However, the previous statement is not valid  in the presence of fermionic fields.

\subsection{\label{CEGF}Cartan-Einstein gravity coupled to fermions}

The Dirac action  can be written in terms of differential forms as
\begin{align}
  S_\Psi ={}& -\int \frac{\epsilon_{\hat{a}_1\cdots \hat{a}_D}}{(D-1)!} \Bps \hvif{\hat{a}_1}\w\cdots\w\hvif{\hat{a}_{D-1}}\ga^{\hat{a}_D}\hat{\cdf}\Psi \notag\\
 & -m\int\frac{\epsilon_{\hat{a}_1\cdots \hat{a}_D}}{D!}\Bps \hvif{\hat{a}_1}\w\cdots\w\hvif{\hat{a}_{D}}\Psi,\label{D-action}
\end{align}
with $\hat{\cdf}$ the exterior derivative twisted by the spin connection (see Eqs. \eqref{eq:covd} and \eqref{eq:cdf}).


Therefore, the equations of motion for the whole system are,
\begin{align}
  \hat{\Ri}^{\hat{m}}{}_{\hat{a}_3} -\frac{1}{2}\hat{\Ri} \delta^{\hat{m}}_{\hat{a}_3} &=\kappa^2\bar{\Psi}\[\ga^{\hat{m}}\hat{\covd}_{\hat{a}_3}-\delta^{\hat{m}}_{\hat{a}_3}\(\fy{\hat{\covd}}+m\)\]\Psi, \label{ED-eom}\\
  \hcont{\hat{a} \hat{b} \hat{c} }{}{} &= -\frac{\kappa^2}{4} \bar{\Psi}\ga_{\hat{a} \hat{b} \hat{c}}\Psi, \label{cont-constraint}
\end{align}
where $\hcont{\hat{a} \hat{b} \hat{c} }{}{}$ is the contorsion, and it is expressed as a function of the torsion by
\begin{equation*}
 \hcont{\hat{a} \hat{b} \hat{c} }{}{} = \frac{1}{2}\(\hat{\tor}_{\hat{b} \hat{c} \hat{a} } + \hat{\tor}_{\hat{b} \hat{a} \hat{c} } + \hat{\tor}_{\hat{a} \hat{b} \hat{c} }\).
\end{equation*}

Equation~\eqref{cont-constraint} is  a constraint and it can be substituted back into the action. It is possible to express $\SPIF{a}{b}$ as the sum of the torsion-free connection and the contribution of the contorsion,
\begin{align}
  \SPIF{a}{b} \mapsto \SPIF{a}{b}+\CONTF{a}{b},
\end{align}
where $\CONTF{a}{b} = \hcont{\hat{\mu}}{\hat{a}}{\hat{b}}\df x^{\hat{\mu}}$.

Substituting this into the total action given by Eqs. \eqref{CE-action} and \eqref{D-action}, one obtains
\begin{align}
  S ={}& \int\de[V\!\!]_D\;\[\frac{1}{2\kappa^2}\hat{\Ri} -\Bps\(\fy{\hat{\covd}}+m\)\Psi\right.\notag\\
  & \left.+\frac{\kappa^2}{32}\bar{\Psi}\ga_{\hat{a} \hat{b} \hat{c}}\Psi\bar{\Psi}\ga^{\hat{a} \hat{b} \hat{c}}\Psi\],\label{action}
\end{align}
which is a torsion-free theory of gravity coupled to a fermion with a four-fermion contact interaction.

In order to compare a model with experimental data, it must contain all of the fields representing the particles of the standard model. Therefore, the whole spectrum of fermions should be added. However, possible modifications of gauge interactions won't be considered. In the case of several fermions, the whole action would be
\begin{align}
  S ={}& \int\de[V\!\!]_D\;\[\frac{1}{2\kappa^2}\hat{\Ri} -\sum_{n\in\text{flav.}} \Bps_n\(\fy{\hat{\covd}}+m\)\Psi_n\right.\notag\\
  & \left.+\sum_{m,n\in\text{flav.}}\frac{\kappa^2}{32}\bar{\Psi}_{m}\ga_{\hat{a} \hat{b} \hat{c}}\Psi_{m}\bar{\Psi}_{n}\ga^{\hat{a} \hat{b} \hat{c}}\Psi_{n}\],\label{mult-action}
\end{align}
where indices $m$ and $n$ represent flavor. Note that coupling constants differs by a factor two depending of whether the four-fermion interaction includes a single or a couple of flavors~\footnote{There are other constructions of torsional gravity coupled with fermions. The interested reader is encouraged to review Ref.~\cite{Fabbri:2011kq,Fabbri:2012yg}}. 


\section{\label{sec:FC}Fermion Condensation}

In this section a model containing the four-fermion interaction in Eq.~\eqref{mult-action} is constructed. It is assumed that the dimensionality of the spacetime is five. Therefore, the effective theory in four dimensions should be found. Although a brief derivation of the effective theory is shown below, a more detailed analysis can be found in  Ref.~\cite{CCSZ}.

The interest in this kind of five-dimensional models has grown recently because they could explain the appearance of quark masses and mixing, induced by the condensation of fermions of a fourth family, whenever a special type of four-fermion interaction term exists~\cite{CarcamoHernandez:2012xy}.

\subsection{Effective theory in four dimensions}

First of all, using the fact that the irreducible representation of the gamma matrices in five and four dimensions are the same, the antisymmetric product $\ga_{\hat{a} \hat{b} \hat{c}}$ is decomposed into
\begin{align}
  \(\ga_{\hat{a} \hat{b} \hat{c}}\)\(\ga^{\hat{a} \hat{b} \hat{c}}\) = \(\ga_{a b c}\)\(\ga^{a b c}\) + 3 \(\ga_{a b *}\)\(\ga^{a b *}\).
\end{align}
Additionally, in the last term the product of gamma matrices can be decomposed further, thus
\begin{align}
  \(\ga_{\hat{a} \hat{b} \hat{c}}\)\(\ga^{\hat{a} \hat{b} \hat{c}}\) &= \(\ga_{a b c}\)\(\ga^{a b c}\) + 3 \(\ga_{a b}\ga_*\)\(\ga^{a b}\ga^*\)\\
  &= 6\(\ga_{a}\ga^*\)\(\ga^{a}\ga^*\) + 3 \(\ga_{a b}\ga^*\)\(\ga^{a b}\ga^*\),
\end{align}
where the definition $\ga^*=\imath\ga^0\ga^1\ga^2\ga^3$ has been used.

Next, using the decomposition of the five-dimensional  fermions in terms of chiral four-dimensional ones,
\begin{align}
  \Psi_m(x,\xi) = f_{m+}(\xi)\;\psi_{m+}(x) +f_{m-}(\xi)\;\psi_{m-}(x),
\end{align}
and the chirality condition  $\ga^*\psi_{m\pm} =\pm\psi_{m\pm}$, the currents involved on Eq.~\eqref{mult-action} are
\begin{align}
  \(J_m\)^{a*} &= \bar{\Psi}_m\ga^{a}\ga^*\Psi_m\notag\\
  &= \abs{f_{m+}}^2 \bps_{m+}\ga^a \psi_{m+} - \abs{f_{m-}}^2 \bps_{m-}\ga^a \psi_{m-}
\end{align}
and
\begin{align}
  \(J_m\)^{ab*}  =& \bar{\Psi}_m\ga^{a b}\ga^*\Psi_m\notag\\
  =& -f^*_{m+} f_{m-} \;\bps_{m+}\ga^{a b} \psi_{m-} \notag\\
  & + {f^*_{m-}}f_{m+}\; \bps_{m-}\ga^{ab} \psi_{m+},
\end{align}
where possible Kaluza-Klein excitations have been dropped. Moreover, in order to evade an overwhelming notation, define~\footnote{Notice that under interchange of chirality ($+\leftrightarrow -$) the quantities change as $a_m\leftrightarrow b_m$ and $c_m\leftrightarrow c_m^*$.}
\begin{align}
  a_m &= \abs{f_{m+}}^2,&  b_m &= \abs{f_{m-}}^2,\\
  c_m &= {f_{m+}^* f_{m-}},&  c^*_m &= {f_{m-}^* f_{m+}}.
\end{align}

Now, using the Fierz identities (see Appendix~\ref{sec:Fierz}) together with the identity for the $SU(N_c)$ generators
\begin{align}
  (T^A)_{ij}(T_A)_{kl}&=\frac{1}{2}\(\delta_{il}\delta_{kj}-\frac{1}{N_c} \delta_{ij}\delta_{kl}\),
\end{align}
the four-fermion interaction terms  yield
\begin{widetext}
  \begin{align}
    \Bps_m\ga^{a}\ga^*\Psi_m\Bps_n\ga_{a}\ga^*\Psi_n &=
    + a_ma_n\(\bps_{m+}\gam{a}\psi_{m+}\)\(\bps_{n+}\ga_{a}\psi_{n+}\) + b_m b_n\(\bps_{m-}\gam{a}\psi_{m-}\)\(\bps_{n-}\ga_{a}\psi_{n-}\) \notag\\
    &\quad - a_m b_n \(\bps_{m+}\gam{a}\psi_{m+}\)\(\bps_{n-}\ga_{a}\psi_{n-}\) - b_m a_n \(\bps_{m-}\gam{a}\psi_{m-}\)\(\bps_{n+}\ga_{a}\psi_{n+}\) 
    \notag\\
    &= + a_ma_n\(\bps_{m+}\gam{a}\psi_{m+}\)\(\bps_{n+}\ga_{a}\psi_{n+}\) 
    \notag\\
    &\quad + 2 a_m b_n \[ 2 \(\bps_{m+} T^A \psi_{n-}\)\(\bps_{n-} T_A \psi_{m+}\) + \frac{1}{N_c} \(\bps_{m+} \psi_{n-}\)\(\bps_{n-} \psi_{m+}\) \]\\
    &\quad + \Big\{ + \leftrightarrow - \Big\}  \notag
  \end{align}
  and
  \begin{align}
    \Bps_m\ga^{ab}\ga^*\Psi_m\Bps_n\ga_{ab}\ga^*\Psi_n &= 
    c_m c_n \(\bps_{m+}\ga^{ab}\psi_{m-}\)\(\bps_{n+}\ga_{ab}\psi_{n-}\) 
    + c_m^* c_n^* \(\bps_{m-}\ga^{ab}\psi_{m+}\)\(\bps_{n-}\ga_{ab}\psi_{n+}\)  \notag\\
    &= 16 c_m c_n \(\bps_{m+} T^A \psi_{n-}\)\(\bps_{n+} T_A \psi_{m-}\) + \frac{8}{N_c}c_m c_n  \(\bps_{m+} \psi_{n-}\)\(\bps_{n+} \psi_{m-}\) \notag\\
    &\quad + 4 c_m c_n \(\bps_{m+} \psi_{m-}\)\(\bps_{n+} \psi_{n-}\) + \Big\{ + \leftrightarrow - \Big\}.
  \end{align}
\end{widetext}

In the following, the discussion will be focused on four-fermion quark-quark interaction given by $\(\bps_{m+}\psi_{n-}\)\(\bps_{n+}\psi_{m-}\)$ terms, because a dynamical symmetry breaking mechanism as that presented by Bardeen {\it et~al.} in Ref.~\cite{PhysRevD.41.1647} is desirable. However, it is worth noticing that in addition to the quark interactions, there are four-lepton interaction and lepton-quark interactions. The former would generate effects as discussed in Ref.~\cite{Burdman:2009ih}, while the latter would emulate lepto-quark interactions and therefore a general model would mimic grand unified theories (GUTs) or supersymmetric scenarios~\footnote{These general aspects of the four-fermion interaction are been considered and will be developed in a future manuscript~\cite{OCF-future}.}.


Therefore, the effective four-fermion action 
obtained from the five-dimensional spacetime is
\begin{widetext}
  \begin{align}
    S[\psi^4] &= \frac{3\kappa^2}{16}\int \de[V\!\!]_{D} \bigg\{a_ma_n\(\bps_{m+}\gam{a}\psi_{m+}\)\(\bps_{n+}\ga_{a}\psi_{n+}\)
    + 2 a_m b_n \[ 2 \(\bps_{m+} T^A \psi_{n-}\)\(\bps_{n-} T_A \psi_{m+}\) + \frac{1}{N_c} \(\bps_{m+} \psi_{n-}\)\(\bps_{n-} \psi_{m+}\) \]\notag\\
    &\phantom{= \frac{3\kappa^2}{16}\int \de[V\!\!]_{D} \bigg\{}
    + 2 c_m c_n \[ 4 \(\bps_{m+} T^A \psi_{n-}\)\(\bps_{n+} T_A \psi_{m-}\) + \frac{2}{N_c}  \(\bps_{m+} \psi_{n-}\)\(\bps_{n+} \psi_{m-}\) +  \(\bps_{m+} \psi_{m-}\)\(\bps_{n+} \psi_{n-}\)\] \notag\\
    &\phantom{= \frac{3\kappa^2}{16}\int \de[V\!\!]_{D} \bigg\{}
     + \Big\{ + \leftrightarrow - \Big\}\bigg\}.
    \label{general4FI}
  \end{align}
\end{widetext}

\subsection{Condensation, masses and mixing}

When two currents $J$ and $J'$ are coupled, it is equivalent to introducing auxiliary fields  {through the substitution}
\begin{align}
 JJ'\mapsto JH'+HJ'-HH',
\end{align}
where the equations of motion for the auxiliary fields are $H=J$ and $H'=J'$. Then, the mean-field approximation can be used, giving $H\approx \vev{J}$ and $H'\approx \vev{J'}$. 

Here, currents have the form
\begin{align}
  J^{\Gamma}=\bar\psi_{mr}\Gamma\psi_{ns}
\end{align}
with $\Gamma=\{1,\gam{a},T^A\}$, $m$ and $n$ the flavor indices, and $r,s$ the chirality indices.
The condensation will pair only the fourth generation of quarks. Since Lorentz and color symmetries must be preserved,  the only allowed condensed current will be with $\Gamma=1$.

The flavor sum on Eq.~\eqref{general4FI} separates into
\begin{align}
  \Lag_{\psi^4}^{(5)} 
  &= 
  \sum_{q,q'}\Lag_{qq'}
  +2\sum_{q,Q}\Lag_{qQ}
  +\sum_{Q, Q'} \Lag_{QQ'},
\end{align}
where $Q,Q'$ represent the fourth quark generation. 
The second term will generate quark masses for the first three generations, with $m_q\sim\kappa^2\vev{\bar{Q}Q}$. 
The last one provides masses for the fourth generation of quarks.

\begin{widetext}
  After condensation, the mass Lagrangian for the first three generations of quarks is
  \begin{equation}
    \Lag_{q^2}^{(5)} =
    \frac{3\kappa^2}{4} c_q\[ c_B\vev{\bar{B}^+ B^-} + c_T\vev{\bar{T}^+ T^-} \] \bar{q}^+ q^- + \Big\{ + \leftrightarrow - \Big\}
  \end{equation}
  
  For the fourth generation of quarks we have
  \begin{align}
    \Lag_{Q^2}^{(5)}  & = 
    \frac{3\kappa^2}{4}\[\frac{a_Tb_T}{N_c}\vev{\bar{T}^+ T^-} + 2\frac{c_Tc_T}{N_c}\vev{\bar{T}^+ T^-} + c_Bc_T \vev{\bar{B}^+ B^-} + c_Tc_T\vev{\bar{T}^+ T^-}\] \bar{T}^+ T^-\notag\\
    &\quad \frac{3\kappa^2}{4}\[\frac{a_Bb_B}{N_c}\vev{\bar{B}^+ B^-} + 2\frac{c_Bc_B}{N_c}\vev{\bar{B}^+ B^-} + c_Bc_T \vev{\bar{T}^+ T^-} + c_Bc_B\vev{\bar{B}^+ B^-}\] \bar{B}^+ B^-\notag\\
    &\quad + \Big\{ + \leftrightarrow - \Big\}.
  \end{align}
\end{widetext}
Assuming that all profiles are real, one might define the coefficients
\begin{align}
  f_{\rho\sigma}=\int_0^R d\xi \frac{\sqrt{\abs{\hat{g}}}}{\sqrt{\abs{g}}} f_+^\rho f_-^\rho f_+^\sigma f_-^\sigma,
\end{align}
the masses of the first three generations of quarks,
\begin{equation}
  m_q = -\frac{3\kappa^2}{4}\left[ f_{qT}\vev{\bar TT} + f_{qB} \vev{\bar BB}  \right],\label{quarkm}
\end{equation}
and the fourth generation quark masses,
\begin{widetext}
  \begin{align}
    m_T &= -\frac{3\kappa^2}{4}\left[\left( 1+\frac{3}{N_c} \right) f_{TT}\vev{\bar TT} + f_{TB}\vev{\bar BB}\right]
    -\frac{g_+^Tg_-^T}{M^2_{KK}} \vev{\bar TT}
    \\
    m_B &=  -\frac{3\kappa^2}{4}\left[\left( 1+\frac{3}{N_c} \right) f_{BB}\vev{\bar BB} + f_{TB}\vev{\bar TT}\right]
    -\frac{g_+^B g_-^B}{M^2_{KK}}\vev{\bar BB},
  \end{align}
\end{widetext}
where the last terms correspond to the exchange of the first Kaluza-Klein gluon mode, with a  mass of $M_{KK}$. 

The inclusion of leptons is straightforward, just adding extra flavors 
singlet of color. Their masses are
\begin{equation}
  m_\ell = -\frac{3\kappa^2}{4}\left[ f_{\ell T}\vev{\bar TT}  + f_{\ell B}\vev{\bar BB} \right].\label{leptonm}
\end{equation}

Note that Eqs.~\eqref{quarkm} and \eqref{leptonm} coincide with the shape of the masses in Ref.~\cite{Burdman:2007sx,CarcamoHernandez:2012xy}. Only the fourth family masses differ due to the $TTBB$ interaction term present in our model.




\section{\label{sec:disc} Discussion and Conclusions}


The developed model has been constructed by considering the quark sector of the standard model coupled to torsionful gravity. As result, a contact four-fermion interaction term appears, coupling at most two different pairs of quarks, providing a natural arena for  symmetry breaking through fermion condensation and, additionally, fermions acquire mass.


A fourth fermion family has been included in order  to condense them, and generate all the wanted features of technicolor, leaving the standard model quarks  outside the condensation scheme. The proposed scenario, as shown above, takes into account a partial contribution to the condensation coming from gravitational torsion, although additional contributions come from the Kaluza-Klein towers.

Due to the special kind of interaction induced by the presence of torsion, the effective mass matrix of fermions is diagonal. This characteristic ensures a simple model, in the sense that no other sources of freedom are involved. Nonetheless, it implies that the  Cabbibo-Kobayashi-Maskawa mass matrix has the same status as that in the standard model.

Additionally, the introduction of extra dimensions is necessary for the gravitational coupling constant $\kappa^2$ to be of order \si{\TeV}$^n$, with $n$ the number of extra dimensions. This serves to ``solve'' the {\it hierarchy problem} and additionally assures that the four-fermion interaction is not suppressed by the four-dimensional Planck's mass $M_{pl}\sim\SI{e19}{\GeV}$, but by a much weaker fundamental gravitational scale, $M_*\sim\si{\TeV}$.

Despite the fact that the considered model does not have the richness of the one presented in Ref.~\cite{CarcamoHernandez:2012xy}, by providing an explanation of the origin of quark masses and mixing, the spectrum of particles is provided by the integration of the profiles along the extra dimension. Since these profiles are usually exponential terms, it can be argued that small differences on the constants that describe them would generate great mass differences, giving a {\it natural} hierarchy on the quark masses. Moreover, due to its simplicity, the model does not require additional symmetries or structures.

In the context of Higgs physics, it is still arguable a composite Higgs with small mass, since fermionic loops represent a negative contribution to the mass of the boson, as a binding energy. This argument is essentially the same as that in walking technicolor models, where the Higgs resonance is around \SI{125}{\GeV} depite the  fact that the   technifermions' masses could be of order \si{\TeV}.

Finally, it is worth  remarking that fermion masses in the proposed scenario are similar to those in previous models. However, the following  differences should be highlighted: (a) This model contains a four-fermion interaction introduced by a minimal generalization of general relativity due to the presence of torsion, (b) no extra symmetries have been imposed on the construction of the model, (c) Naturally, fermions are paired in the extra interaction, and the quark mixing keeps the status as in the standard model, and (d) although this model starts with a different current structure (compared with the mentioned models), the effective theory has the same physical terms; therefore,  condensation for this model is assured by the conditions on Ref.~\cite{Burdman:2007sx,CarcamoHernandez:2012xy}.

At the LHC, bounds to the four-fermion interaction term have been found (See Ref.~\cite{Aad:2011aj,ATLAS:2012pu,Aad:2012bsa}), typically $\Lambda\sim\SI{ 10}{\TeV}$. Additionally, there exist cosmological constraints, (See Ref.~\cite{Chang:2000yw}), where $\Lambda\sim \SI{28}{\TeV}$.
However, these constraints are in four dimensions, while our model has one extra dimension. Since the parameters of the theory depend on the particular construction, no universal constraints can be imposed.
Nonetheless, in a previous report, we found some constraints to the size of the extra dimensions  of the spacetime~\cite{CCSZ}.

\section*{Acknowledgements}

We would like to thank  A. Toloza, A. Anabalon, S. Kovalenko, and I. Schmidt for helpful discussions and comments. This work was partially supported  by  Conicyt (Chile) under Grant No. 21130179 and  Basal Project FS0821, and Fondecyt Projects No. 1120346 and No. 11000287.

\appendix

\section{\label{app:notation}CONVENTIONS AND NOTATIONS}

\subsection*{Spacetime and metric}
Throughout  the paper, the metrics have a  signature that is  mostly positive. 
Since the vielbein formalism is used extensively, the distinction between flat and curved coordinates is through latin and greek indices, respectively.   
Moreover, hatted indices run over the whole spacetime (say five-dimensional spacetime) while unhatted ones run over a hypersurface restriction, i.e., a four-dimensional spacetime. 

The vielbein formalism relies on the definition of a Lorentzian frame at each point of the spacetime through the relation
\begin{align*}
  \hat{g}_{\hat{\mu}\hat{\nu}} = \eta_{\hat{a} \hat{b}}\VI{a}{\mu}\VI{b}{\nu},
\end{align*}
where $\VI{a}{\mu}$ are the vielbeins, and they encode the geometric information of the curved spacetime when one ``translates'' into the tangent space. 
These objects are invertible, and their inverse are denoted by 
\begin{align*}
  \VIN{\mu}{a} \equiv \(\VI{a}{\mu}\)^{-1}.
\end{align*}


\subsection*{Clifford algebra and spinors}
The gamma matrices are defined on the tangent space, and they satisfy the Clifford algebra,
\begin{equation}
  \left\{\GAM{a},\GAM{b}\right\}=2\eta^{\hat{a}\hat{b}}\mathds{1}.\label{Cliff-alg}
\end{equation}

For the sake of clarity in the following, spacetimes are considered five and four dimensional. Thus, hatted indices run over $\hat{a}= 0,...,4$, while unhatted ones run over $a=0,...,3$.

In even dimensions one can define the chirality matrix $\ga^*$, satisfying the properties
\begin{align*}
  \left\{\gam{a},\ga^* \right\}=0,\quad (\ga^*)^2=\mathds{1},
\end{align*}
and the (chiral) projector operators,
\begin{equation}
  P_\pm= \frac{\mathds{1}-\ga^*}{2},
\end{equation}
are both nontrivial.

On the other hand, odd-dimensional Clifford algebras are constructed by using the gamma matrices of the codimension-one spacetime, via 
\begin{align}
  \GAM{a} = \(\gam{b},\ga^*\).\label{GAMa}
\end{align}
These odd-dimensional Clifford algebras have  trivial projectors $P_\pm$, and therefore chiral fermions cannot be defined.

In any dimension one may define a set of generators of the Lorentz algebra, constructed with the elements of the Clifford algebra (\ref{Cliff-alg}). These generators of the Lorentz algebra are
\begin{equation}
  \J^{\hat{a}\hat{b}}=-\frac{\imath}{4}\left[\GAM{a},\GAM{b}\right],
\end{equation}
which are known as the generators in the spin representation.

In curved spacetime the Dirac equation is obtained by replacing the partial derivative by the  derivative twisted by the spin connection,
\begin{align}
  \PA{\mu}\to \hat{\covd}_{\hat{\mu}}= \PA{\mu}+\frac{\imath}{2}\(\SPI{\mu}\)^{\hat{a}\hat{b}}\J_{\hat{a}\hat{b}},\label{eq:covd}
\end{align}
which defines the exterior derivative twisted operator by
\begin{align}
  \hat{\cdf} &= \[\PA{\mu}+\frac{\imath}{2}\(\SPI{\mu}\)^{\hat{a}\hat{b}}\J_{\hat{a}\hat{b}}\]\df x^\mu\notag\\
  &= \df +\frac{\imath}{2}\hspif{\hat{a}\hat{b}}{}\J_{\hat{a}\hat{b}}.\label{eq:cdf}
\end{align}

The Dirac-Feynman slash notation  must be interpreted as
\begin{equation}
  \fy{\hat{\covd}}=\VIN{\mu}{a}\GAM{a}\hat{\covd}_{\hat{\mu}},
\end{equation}
with $\VIN{\mu}{a}$  the inverse of the vielbein $\VI{a}{\mu}$.

\section{\label{sec:actions}EQUIVALENCY OF ACTIONS}

In this appendix the equivalency between the actions in Eqs. \eqref{CE-action} and \eqref{D-action} with their best known form 
\begin{align}
  S_{gr} &= \frac{1}{2\kappa^2}\int\de[V\!\!]_D\; \Ri,\\
  S_\Psi &= -\int\de[V\!\!]_D\; \Bps\(\fy\covd+m\)\Psi,
\end{align}
is shown.

In order to achieve the goal, one needs a couple of identities which follow from the signature of the Lorentz metric and the orientability of the spacetime and the usual $\epsilon$ identities,
\begin{align}
  \hvif{\hat{a}_1}\w\cdots\w\hvif{\hat{a}_D}&= -\epsilon^{\hat{a}_1\cdots\hat{a}_D}\de[V\!\!]_D,\label{viel-id}\\
  \epsilon^{\hat{a}_1\hat{a}_2\cdots\hat{a}_D}\epsilon_{\hat{a}_1\hat{a}_2\cdots\hat{a}_D} &= -(D)!,\label{eps-cont}\\
  \epsilon^{\hat{m}\hat{a}_2\cdots\hat{a}_D}\epsilon_{\hat{n}\hat{a}_2\cdots\hat{a}_D} &= -(D-1)!\delta^{\hat{m}}_{\hat{n}},\label{eps-1}\\
  \epsilon^{\hat{m}_1\hat{m}_2\hat{a}_3\cdots\hat{a}_D}\epsilon_{\hat{n}_1\hat{n}_2\hat{a}_3\cdots\hat{a}_D} &= -(D-2)!\tensor*{\delta}{*^{\hat{m}_1}_{\hat{n}_1}^{\hat{m}_2}_{\hat{n}_2}},\label{eps-2}
\end{align}
where
\begin{align*}
  \tensor*{\delta}{*^{\hat{m}_1}_{\hat{n}_1}^{\hat{m}_2}_{\hat{n}_2}}= \delta^{\hat{m}_1}_{\hat{n}_1}\delta^{\hat{m}_2}_{\hat{n}_2} - \delta^{\hat{m}_1}_{\hat{n}_2}\delta^{\hat{m}_2}_{\hat{n}_1}.
\end{align*}

\begin{Thm}
  \begin{align}
   \frac{\epsilon_{\hat{a}_1\cdots \hat{a}_D}}{(D-2)!}\hRif{\hat{a}_1 \hat{a}_2}{}\w\hvif{\hat{a}_3}\w\cdots\w\hvif{\hat{a}_D}= \de[V\!\!]_D\;\Ri
  \end{align}
\end{Thm}
\begin{proof}
  Start by writing the curvature 2-form in the vielbein basis,
  \begin{align*}
    \hRif{\hat{a}_1 \hat{a}_2}{} = \frac{1}{2}\hat{\Ri}^{\hat{a}_1 \hat{a}_2}{}_{\hat{m} \hat{n}}\VIF{m}\w\VIF{n},
  \end{align*}
then from Eq. \eqref{viel-id}, one gets
\begin{align*}
  \hRif{\hat{a}_1 \hat{a}_2}{}\w\hvif{\hat{a}_3}\w\cdots\w\hvif{\hat{a}_D} = -\frac{1}{2}\hat{\Ri}^{\hat{a}_1 \hat{a}_2}{}_{\hat{m} \hat{n}}\epsilon^{\hat{m}\hat{n}\hat{a}_3\cdots\hat{a}_D}\de[V\!\!]_D.
\end{align*}
Contracting the last expression with $\epsilon_{\hat{a}_1\cdots \hat{a}_D}$ yields
\begin{align*}
  \epsilon_{\hat{a}_1\cdots \hat{a}_D}\hRif{\hat{a}_1 \hat{a}_2}{}\w\hvif{\hat{a}_3}\w\cdots\w\hvif{\hat{a}_D} 
  = (D-2)! \;\Ri\;\de[V\!\!]_D,
\end{align*}
which ends the proof.
\end{proof}

\begin{Thm}
  \begin{align}
    \de[V\!\!]_D\; \Bps\(\fy\covd+m\)\Psi ={}& \frac{\epsilon_{\hat{a}_1\cdots \hat{a}_D}}{(D-1)!} \Bps \hvif{\hat{a}_1}\w\cdots\w\hvif{\hat{a}_{D-1}}\ga^{\hat{a}_D}\hat{\cdf}\Psi \notag\\
    & +m\int\frac{\epsilon_{\hat{a}_1\cdots \hat{a}_D}}{D!}\Bps \hvif{\hat{a}_1}\w\cdots\w\hvif{\hat{a}_{D}}\Psi.
  \end{align}
\end{Thm}
\begin{proof}
  The proof is split in two parts. 

  First, consider the mass term. Using the identities in Eqs. \eqref{viel-id} and \eqref{eps-cont}, it follows that
  \begin{align}
    \frac{\epsilon_{\hat{a}_1\cdots \hat{a}_D}}{D!}\hvif{\hat{a}_1}\w\cdots\w\hvif{\hat{a}_{D}} = \de[V\!\!]_D ,
  \end{align}
  which ensures that the mass terms of both sides are equal.
  
  Next is the kinetic term. One begins by expanding the exterior covariant derivative in the vielbein basis,
  \begin{align}
    \hat{\covd} = \hat{\covd}_{\hat{m}}\VIF{m}.
  \end{align}
  Then, using the identities in Eqs. \eqref{viel-id} and \eqref{eps-1}, one gets
  \begin{align}
    \tfrac{\epsilon_{\hat{a}_1\cdots \hat{a}_D}}{(D-1)!}  \hvif{\hat{a}_1}\w\cdots\w\hvif{\hat{a}_{D-1}}\w\VIF{m}\ga^{\hat{a}_D}\hat{\cdf}_{\hat{m}} &= \de[V\!\!]_D\;\delta^{\hat{m}}_{\hat{a}_D}\ga^{\hat{a}_D}\hat{\cdf}_{\hat{m}}\notag\\
    &= \de[V\!\!]_D\;\fy{\hat{\cdf}}.
  \end{align}

  Therefore, the action in Eq. \eqref{D-action} is in fact the usual Dirac action written in the language of differential forms.
\end{proof}

\section{\label{sec:Fierz}FIERZ IDENTITIES}

Throughout the paper, the spacetime dimension has been set to be either five or four. Since irreducible representations of the Clifford algebra on both have the same dimension, the Fierz identities are equal, and coincide with the ones stated in the usual text books on quantum field theory, such as Ref.~\cite{Langacker}.

Below, the identities used in Sec.~\ref{sec:FC} are shown without proof:
\begin{widetext}
  \begin{align}
    \(\bps_{1-}\gam{a}\psi_{2-}\)\(\bps_{3-}\ga_{a}\psi_{4-}\)  &=    \(\bps_{1-}\gam{a}\psi_{4-}\)\(\bps_{3-}\ga_{a}\psi_{2-}\),\\
    \(\bps_{1+}\gam{a}\psi_{2+}\)\(\bps_{3+}\ga_{a}\psi_{4+}\)  &=    \(\bps_{1+}\gam{a}\psi_{4+}\)\(\bps_{3+}\ga_{a}\psi_{2+}\),\\
    \(\bps_{1+}\gam{a}\psi_{2+}\)\(\bps_{3-}\ga_{a}\psi_{4-}\)  &= -2 \(\bps_{1+}\psi_{4-}\)\(\bps_{3-}\psi_{2+}\),\\
    \(\bps_{1+}\psi_{2-}\)\(\bps_{3+}\psi_{4-}\)                &= -\frac{1}{2} \(\bps_{1+}\psi_{4-}\)\(\bps_{3+}\psi_{2-}\) +\frac{1}{8} \(\bps_{1+}\gam{ab}\psi_{4-}\)\(\bps_{3+}\ga_{ab}\psi_{2-}\),\notag\\
    \(\bps_{1+}\gam{ab}\psi_{2-}\)\(\bps_{3-}\ga_{ab}\psi_{4+}\) &= 0.
  \end{align}
\end{widetext}


\end{document}